\documentclass[journal]{IEEEtran}

\ifCLASSINFOpdf
\else
\fi
\usepackage{geometry}
\usepackage{hyperref}
\usepackage{courier}
\usepackage{times}
\usepackage{amsmath,amssymb}
\usepackage{enumitem}
\usepackage[boxed,vlined,longend]{algorithm2e}
\usepackage{color}
\usepackage[normalem]{ulem}

\newtheorem{theorem}{Theorem}
\newtheorem{proposition}[theorem]{Proposition}

\newtheorem{corollary}[theorem]{Corollary}
\newtheorem{lemma}[theorem]{Lemma}
\newenvironment{proof}
{\begin{trivlist}\item[]{{\sc Proof.}}}{\hfill{$\square$}\noindent\end{trivlist}}

\def\llceil{\lceil\kern-3.2pt\lceil}
\def\rrceil{\rceil\kern-3.2pt\rceil}
\def\llfloor{\lfloor\kern-3.2pt\lfloor}
\def\rrfloor{\rfloor\kern-3.2pt\rfloor}
\def\leftllceil{\left\lceil\kern-3.2pt\left\lceil}
\def\rightrrceil{\right\rceil\kern-3.2pt\right\rceil}
\def\leftllfloor{\left\lfloor\kern-3.2pt\left\lfloor}
\def\rightrrfloor{\right\rfloor\kern-3.2pt\right\rfloor}
\def\bigllfloor{\bigg\lfloor\kern-3.2pt\bigg\lfloor}
\def\bigrrfloor{\bigg\rfloor\kern-3.2pt\bigg\rfloor}

\begin{document}


\title{No projective $16$-divisible binary linear code of length $131$ exists}

\author{Sascha Kurz, University of Bayreuth
\thanks{S.~Kurz is with the Department of Mathematics, Physics, and 
Computer Science, University of Bayreuth, Bayreuth, GERMANY. email: sascha.kurz@uni-bayreuth.de}}

%



\maketitle

\begin{abstract}
  We show that no projective $16$-divisible binary linear code of length $131$ exists. This implies several improved upper bounds for 
  constant-dimension codes, used in random linear network coding, and partial spreads.
\end{abstract}


%
\IEEEpeerreviewmaketitle

\maketitle

\section{Introduction}
\label{sec_intro}
\IEEEPARstart{A}{n} $[n,k,d]_q$ code is a $q$-ary linear code with length $n$, dimension $k$, and minimum Hamming distance $d$. Since we will 
only consider binary codes, we also speak of $[n,k,d]$ codes. Linear codes have numerous applications so that constructions or non-existence 
results for specific parameters were the topic of many papers. One motivation was the determination of the smallest integer $n(k,d)$ for which 
an $[n,k,d]$ code exists. As shown in \cite{baumert1973note} for every fixed dimension $k$ there exists an integer $D(k)$ such that $n(k,d)=g(k,d)$ 
for all $d\ge D(k)$, where     
  $n(k,d)\ge g(k,d):=\sum_{i=0}^{k-1} \left\lceil\frac{d}{2^i}\right\rceil$,
is the so-called Griesmer bound. Thus, the determination of $n(k,d)$ is a finite problem. In 2000 the determination of $n(8,d)$ was completed in 
\cite{bouyukhev2000smallest}. Not many of the open cases for $n(9,d)$ have been resolved since then and we only refer to most recent paper \cite{46_9_20}.

The aim of this note is to to circularize a recent application of non-existence results of linear codes. In random linear network coding so-called 
constant-dimension codes are used. These are sets of $k$-dimensional subspaces of $\mathbb{F}_q^n$ with subspace distance $d_S(U,W):=\dim(U)+\dim(W)-
2\dim(U\cap W)$. By $A_q(n,d;k)$ we denote the maximum possible cardinality, where $A_q(n,d;k)=A_q(n,d;n-k)$, so that we assume $2k\le n$. In 
\cite{kiermaier2020lengths} the upper bounds 
  $A_q(n,d;k)
	  \le \leftllfloor\frac{\left(q^n-1\right)\cdot A_q(n-1,d;k-1)/(q-1)}{\left(q^k-1\right)/(q-1)}\rightrrfloor_{q^{k-1}}$  
for $d>2k$ and $A_q(n,2k;k)\le \leftllfloor\frac{\left(q^n-1\right)/(q-1)}{\left(q^k-1\right)/(q-1)}\rightrrfloor_{q^{k-1}}$ were proven. 
Here $\leftllfloor a/b \rightrrfloor_{q^{r}}$ denotes the maximal integer $t$ such that there exists a 
$q^r$-divisible $q$-ary linear code of effective length $n=a-tb$ and a code is called $q^r$-divisible if the Hamming weights $\operatorname{wt}(c)$ of 
all codewords $c$ are divisible by $q^r$. For integers $r$ the possible length of $q^r$-divisible codes have been completely determined in \cite{kiermaier2020lengths} and except 
for the cases $(n,d,k,q)=(6,4,3,2)$ and $(8,4,3,2)$ no tighter bound for $A_q(n,d;k)$ with $d>2k$ is known. For the case $d=2k$, where the constant-dimension codes 
are also called partial spreads, the notion of 
$\leftllfloor a/b \rightrrfloor_{q^{r}}$ can be sharpened by requiring the existence of a projective $q^r$-divisible $q$-ary linear code of 
effective length $n=a-tb$. Doing so, all known upper bounds for $A_q(n,2k;k)$ follow from non-existence results of projective $q^r$-divisible 
codes, see e.g.\ \cite{honold2018partial}. For each field size $q$ and each integer $r$ there exists only a finite set $\mathcal{E}_q(r)$ 
such that there does not exist a projective $q^r$-divisible code of effective length $n$ iff $n\in \mathcal{E}_q(r)$. We have $\mathcal{E}_2(1)=\{1,2\}$, 
$\mathcal{E}_2(2)=\{1,2,3,4,5,6,9,10,11,12,13\}$, and remark that the determination of $\mathcal{E}_2(3)$ was recently completed in \cite{honold2019lengths} 
by excluding length $n=59$.

In this paper we show the non-existence of $16$-divisible binary codes of effective length $n=131$, which e.g.\ implies $A_2(13,10;5)\le 259$. 

\section{Preliminaries}
\label{sec_preliminaries}
Since the minimum Hamming distance is not relevant in our context, we speak of $[n,k]$ codes. The dual code of an $[n,k]$ code $C$ is the $[n,n-k]$ code  
$C^*$ consisting of the elements of $\mathbb{F}_2^n$ that are perpendicular to all codewords of $C$. By $a_i$ we denote the number of codewords 
of $C$ of weight $i$. With this, the weight enumerator is given by $W(z)=\sum_{i\ge 0} a_iz^i$. The numbers $a_i^*$ of codewords of 
the dual code of weight $i$ are related by the so-called MacWilliams identities
\begin{equation}
  \label{eq_macwilliams}
    \sum_{i\ge 0} a_i^* z^i=\frac{1}{2^k}\cdot \sum_{i\ge 0} a_i(1+z)^{n-i}(1-z)^i.
\end{equation} 
Clearly we have $a_0=a_0^*=1$. 
In this paper we assume that all lengths are equal to the so-called effective length, i.e., $a_1^*=0$. A linear code is called 
projective if $a_2^*=0$. Let $C$ be a projective $[n,k]$ code. By comparing the coefficients of $z^0$, $z^1$, $z^2$, and $z^3$ 
on both sides of Equation~\ref{eq_macwilliams} we obtain:
\begin{eqnarray}
  \sum_{i>0} a_i &=& 2^k-1,\label{mw0}\\
  \sum_{i\ge 0} ia_i &=& 2^{k-1}n,\label{mw1}\\
  \sum_{i\ge 0} i^2a_i &=& 2^{k-1}\cdot n(n+1)/2,\label{mw2}\\ 
  \sum_{i\ge 0} i^3a_i &=& 2^{k-2}\cdot \left(\frac{n^2(n+3)}{2}-3a_3^*\right)\label{mw3}
\end{eqnarray} 

The weight enumerator of a linear $[n,k]$ code $C$ can be refined to a so-called partition weight enumerator, see e.g.\ \cite{simonis1995macwilliams}. 
To this end let $r\ge 1$ be an integer and $\cup_{j=1}^r P_j$ be a partition of the coordinates $\{1,\dots,n\}$. By 
$I=\left(i_1,\dots,i_r\right)$ we denote a multi-index, where $0\le i_j\le p_j$ and $p_j=\#P_j$ for all $1\le j\le r$. With this, 
$a_I\in\mathbb{N}$ denotes the number of codewords $c$ such that $\#\left\{h\in P_j\,:\, c_h\neq 0\right\}=i_j$ for all $1\le j\le r$, which generalizes 
the notion of the counts $a_i$. By $a_I^*\in\mathbb{N}$ we denote the corresponding counts for the dual code $C^*$ of $C$. The generalized 
relation between the $a_I^*$ and the $a_I$ is given by: 
\begin{eqnarray}
   \!\!\!\!\!\!\!&& \sum_{I=\left(i_1,\dots,i_r\right)} a_I^* \prod_{j=1^r} z_j^{i_j} \notag\\ 
    \!\!\!\!\!\!\!&=&\!\!\frac{1}{2^k}\cdot \!\!\!\!\sum_{I=\left(i_1,\dots,i_r\right)}\!\!\!\!\!\! a_I\prod_{j=1}^r \left(1+z_j\right)^{n-i_j}\left(1-z_j\right)^{i_j}\label{eq_macwilliams_partition}
\end{eqnarray}  

The support $\operatorname{supp}(c)$ of a codeword $c\in\mathbb{F}_2^n$ is the set of coordinates 
$\left\{1\le i\le n\,:\, c_i\neq 0\right\}$. The residual of a linear code $C$ with respect of a codeword $c\in C$ 
is the restriction of the codewords of $C$ to those coordinates that are not in the support of $c$, i.e., the 
resulting effective length is given by $n-\operatorname{wt}(c)$. If $c$ is a codeword of a $q^r$-divisible 
$q$-ary code $C$, where $r\ge 1$, then the residual code with respect to $c$ is $q^{r-1}$-divisible, see e.g.\ 
\cite{honold2018partial}. The partition weight enumerator with respect to a codeword $c$ is given by 
Equation~(\ref{eq_macwilliams_partition}), where we choose $r=2$, $P_2=\operatorname{supp}(c)$, and 
$P_1=\{1,\dots,n\}\backslash P_2$, so that restricting to the coordinates in $P_1$ gives the residual code.       
                     
\section{No projective $16$-divisible binary linear code of length $131$ exists}   
\label{sec_main_result}

Assume that $C$ is a projective $16$-divisible $[131,k]$ code. Since for every codeword $c\in C$ the residual code is 
$8$-divisible and projective, we conclude from $\{3,19,35\}\subseteq \mathcal{E}_2(3)$, see e.g.\ \cite{honold2019lengths}, 
that the possible non-zero weights of the codewords in $C$ are contained in $\{16,32,48,64,80\}$. For codewords 
of weight $80$ the weight enumerator of the corresponding residual code can be uniquely determined:
\begin{lemma}
  \label{lemma_8_div_n_51}(\cite[Lemma 24]{honold2018partial})\\
  The weight enumerator of a projective $8$-divisible binary linear code of (effective) length $n=51$ is given by 
  $W(z)=1+204z^{24}+51z^{32}$, i.e., it is an $8$-dimensional two-weight code.
\end{lemma} 

\begin{lemma}
  \label{lemma_general_information}
  Each projective $16$-divisible $[131,k]$ code satisfies
  \begin{eqnarray*}
  a_{48} &=& -6 a_{16} - 3 a_{32} - 10 + 11\cdot 2^{k-9},\\
  a_{64} &=& 8 a_{16} + 3 a_{32} + 15 +221\cdot 2^{k-8},\\
  a_{80} &=& -3 a_{16} - a_{32} - 6 + 59\cdot 2^{k-9},\\
  a_3^* &=& 2^{17-k} a_{16}+2^{15-k}a_{32} -311+5\cdot 2^{16-k},
  \end{eqnarray*}
  $k\ge 9$, and $a_{80}\ge 4+3\cdot 2^{k-5}\ge 52$.
\end{lemma}
\begin{proof}
  Solving the constraints (\ref{mw0})-(\ref{mw3}) for $a_{48}$, $a_{64}$, $a_{80}$, and $a_3^*$ gives 
  the stated equations for general dimension $k$. Since $a_{48}\in\mathbb{N}$ (or $a_{80}\in\mathbb{N}$) 
  we have $k\ge 9$. Since $a_{48}\ge 0$, we have $6 a_{16} + 3 a_{32}\le 11\cdot 2^{k-9}-10$, so that 
  $a_{80} = -3 a_{16} - a_{32} - 6 + 59\cdot 2^{k-9} \ge 4+3\cdot 2^{k-5}\ge 52$. 
\end{proof}

First we exclude the case of dimension $k=9$:
\begin{lemma}
  \label{lemma_exclude_k_9}
  No projective $16$-divisible $[131,9]$ code exists.
\end{lemma}
\begin{proof}
  For $k=9$ the equations of Lemma~\ref{lemma_general_information} yield
  \begin{eqnarray*}
    a_{48} &=& -6 a_{16} - 3 a_{32} + 1,\\ 
    a_{64} &=& 8 a_{16} + 3 a_{32} + 457,\\ 
    a_{80} &=& -3 a_{16} - a_{32} + 53, \text{ and}\\ 
    a_3^*  &=& 256 a_{16} + 64 a_{32} + 329
  \end{eqnarray*}
  for a projective $16$-divisible $[131,9]$ code $C$. 
  Since $a_{48}\ge 0$ and $a_{16},a_{32}\in \mathbb{N}$, we have $a_{16}=a_{32}=0$, so that $a_{48}=1$, 
  $a_{64}=457$, $a_{80}=53$, and $a_3^*=329$. Now consider a codeword $c_{80}\in C$ of weight $80$ and the unique 
  codeword $c_{48}\in C$ of weight $48$. In the residual code of $c_{80}$ the restriction of $c_{48}$ has weight $24$ or $32$ 
  due to Lemma~\ref{lemma_8_div_n_51}. In the latter case the codeword $c_{80}+c_{48}\in C$ has weight $96$, which cannot 
  occur in a projective $16$-divisible binary linear code of length $131$. Thus, we have that $c_{80}+c_{48}\in C$ gives another 
  codeword of weight $80$. However, since $a_{80}$ is odd, this yields a contradiction and the code $C$ does not exist.  
\end{proof}  

\begin{lemma}
  \label{lemma_weight_16_or_32}
  A projective $16$-divisible binary linear code $C$ of length $131$ does not contain a codeword of weight $16$ or $32$. 
\end{lemma}
\begin{proof}
  Let $c\in C$ be an arbitrary codeword of weight $80$ (which indeed exists, see Lemma~\ref{lemma_general_information}) and 
  $c'\in C$ a codeword of weight $16$ or $32$. We consider the  
  residual code $C'$ of $C$ with respect to the codeword $c$. From Lemma~\ref{lemma_8_div_n_51} we conclude that the restriction 
  $\tilde{c}'$ of $c'$ in $C'$ has weight $0$, $24$, or $32$. Since $c+c'\in C$ has a weight of at most $80$, $\tilde{c}'$ is the 
  zero codeword of weight $0$. In other words, we have $\operatorname{supp}(c')\subseteq\operatorname{supp}(c)$. If $L$ denotes 
  the set of codewords of weight $80$ in $C$, then 
  $
    \operatorname{supp}(c') \subseteq \cap_{l\in L} \operatorname{supp}(l)=:M
  $, 
  with $M\subseteq\{1,\dots 131\}$ and $\#M\ge 16$.
  
  Now let $D$ be the code generated by the elements in $M$, i.e., the codewords of weight $80$. By $k'$ we denote the dimension of 
  $D$ and by $k$ the dimension of $C$. 
  Since $D$ contains all codewords of weight $80$ and due to Lemma~\ref{lemma_general_information} we have
  \begin{equation}
    \label{ie_lb_ub}
    4+3\cdot 2^{k-5}\le a_{80}\le 2^{k'}-1 
  \end{equation}
  for $C$. Since $\#M\ge 16$ each generator matrix $G$ of $D$ contains a column that occurs at least $16$ times, i.e., the maximum column 
  multiplicity is at least $16$. If a row is appended to $G$ then the maximum column multiplicity can go down by a factor of at most 
  the field size $q$, i.e., $2$ in our situation. Thus, we have $k'\le k-4$. Since Inequality~\ref{ie_lb_ub}) gives  
  $$
    4+3\cdot 2^{k-5} \le 2^{k'}-1\le 2^{k-4}-1, 
  $$       
  we obtain a contradiction. Thus, we conclude $a_{16}=a_{32}=0$.
\end{proof}

\begin{theorem}
  \label{main_thm}
  No projective $16$-divisible binary linear code of length $131$ exists.  
\end{theorem}
\begin{proof}
  Assume that $C$ is a projective $16$-divisible $[131,k]$ code. From Lemma~\ref{lemma_weight_16_or_32} we conclude 
  $a_{16}=a_{32}=0$, so that Lemma~\ref{lemma_general_information} yields $a_3^*=5\cdot 2^{16-k}-311$. Note that for $k\ge 11$ 
  the non-negative integer $a_3^*$ would be negative. The case $k=9$ is excluded in Lemma~\ref{lemma_exclude_k_9}. In the 
  remaining case $k=10$ we have $a_3^*=9$ and $a_{80}=112$.
  
  Now consider the residual code $C'$ of $C$ with respect to a codeword $c$ of weight $80$. Plugging in the weight enumerator for 
  $C'$ from Lemma~\ref{lemma_8_div_n_51} in Equations~(\ref{mw0})-(\ref{mw3}) gives $a_3^*(C')=17$. Thus, we conclude $a_3^*(C)\ge 17$, 
  which is a contradiction.
\end{proof}

We remark that some parts of our argument can be replaced using the partition weight enumerator from Equation~(\ref{eq_macwilliams_partition}). 
If we consider the partition weight enumerator with respect to a codeword $c$ of weight $80$, then we have $r=2$, $p_1=51$, and $p_2=80$. The 
possible indices where $a_I$ might be positive are given by $(0,0)$, $(0,16)$, $(0,32)$, $(0,48)$, $(0,64)$, $(0,80)$, $(24,24)$, 
$(24,40)$, $(24,56)$, $(32,32)$, and $(32,48)$. Clearly, we have $a_{(0,0)}=1$ and $a_{(0,80)}=1$. By considering the sums of a codeword 
with $c$ we conclude $a_{(0,16)}=a_{(0,64)}$, $a_{(0,32)}=a_{(0,48)}$, $a_{(24,24)}=a_{(24,56)}$, and $a_{(32,32)}=a_{(32,48)}$. From 
Lemma~\ref{lemma_8_div_n_51} we conclude $a_{(32,32)}=a_{(32,48)}=51\cdot 2^{k-9}$, $a_{(24,24)}=a_{(24,56)}=t$, and 
$a_{(24,40)}=204\cdot 2^{k-8}-2t$, where $k$ is the dimension of the code and $t\in \mathbb{N}$ a free parameter. Plugging into 
Equation~(\ref{eq_macwilliams_partition}) this gives $a_{(0,16)}+a_{(0,32)}=2^{k-9}-1$ for the coefficients of $t_1^0t_2^0$ since $a_{(0,0)}^*=1$. 
Using this equation automatically gives $a_{(1,0)}^*=0$, $a_{(2,0)}^*=0$, and $a_{(3,0)}^*=17$. Since $a_{(0,2)}^*=0$ the coefficient 
of $t_2^2$ gives $6320-7344\cdot 2^{k-9}+1024t+ 2224a_{(0,16)} + 176a_{(0,32)}=0$. Thus, we have $a_{(0,16)}=7\cdot 2^{k-10}-3-\tfrac{t}{2}$ and 
$a_{(0,32)}=2-5\cdot 2^{k-10}+\tfrac{t}{2}$. The coefficient of $t_1^1t_2^2$ then gives $a_{(1,2)}^*=408-3t\cdot 2^{14-k}$. For $k=9$ the non-negativity conditions 
$a_{(0,16)},a_{(0,32)}\ge 0$ force $t=1$, so that $a_{(0,0)}=1$, $a_{(0,16)}=a_{(0,64)}=0$, $a_{(0,32)}=a_{(0,48)}=0$, $a_{(0,80)}=1$, 
$a_{(24,40)}=406$, $a_{(24,24)}=a_{(24,56)}=1$, and $a_{(32,32)}=a_{(32,48)}=51$.      
It can be checked that all coefficients on the right hand side of Equation~(\ref{eq_macwilliams_partition}) are non-negative. 
$a_{(0,32)}\ge 0$ implies $t\ge 5\cdot 2^{k-9}-4$, so that $a_{(1,2)}^*$ would be negative for $k\ge 12$. 

Theorem~\ref{main_thm} implies a few further results.
\begin{proposition}
  For $t\ge 0$ we have $A_2(8+5t,10;5)\le 3+2^8\cdot \frac{32^t-1}{31}$. 
\end{proposition}
\begin{proof}
  Assume that $\mathcal{C}$ is a set of $4+2^8\cdot \frac{32^t-1}{31}$ $5$-dimensional subspaces in $\mathbb{F}_2^{8+5t}$ with pairwise trivial 
  intersection. Then, the number of vectors in $\mathbb{F}_2^{8+5t}$ that are disjoint to the vectors of the elements of $\mathcal{C}$ is given 
  by $\left(2^{8+5t}-1\right)-31\cdot \left(4+2^8\cdot \frac{32^t-1}{31}\right)=131$. Thus, by \cite[Lemma 16]{honold2018partial}, there exists 
  a projective $2^{5-1}$-divisible binary linear code of length $n=131$, which contradicts Theorem~\ref{main_thm}.
\end{proof}
The recursive upper bound for constant-dimension codes mentioned in the introduction implies:
\begin{corollary}
  We have $A_2(14,10;6)\le 67\,349$, $A_2(15,10;7)\le 17\,727\,975$, 
  and $A_2(19,10,6)\le 70\,329\,353$.
\end{corollary}
As an open problem we mention that the non-existence of a projective $16$-divisible binary linear code of length $n=130$ would imply $A_2(15,12;6)\le 514$.

\begin{lemma}
  \label{lemma_aux_1}
  For $k\ge 1$, $r\ge 3$, and $j\le 2r-1$ no projective $2^r$-divisible $\left[3+j\cdot 2^{r},k\right]$ code exists.  
\end{lemma}
\begin{proof}
  In \cite[Theorem 12]{honold2018partial} it was proven that the length $n$ of a projective $2^r$-divisible binary linear code either 
  satisfies $n>r2^{r+1}$ or can be written as $n=a\left(2^{r+1}-1\right)+b2^{r+1}$ for some non-negative integers $a$ and $b$. Using $r\ge 3$, we 
  note that $3+j\cdot 2^{r}\le 3+(2r-1)\cdot 2^{r}= 3-2^r+r2^{r+1}<r2^{r+1}$. If $a\left(2^{r+1}-1\right)+b2^{r+1}=3+j\cdot 2^{r}$, then $3+a$ is 
  divisible by $2^r$, so that $a\ge 2^r-3$. However, for $r\ge 3$ we have $a\left(2^{r+1}-1\right)+b2^{r+1}\ge \left(2^r-3\right)\cdot\left(2^{r+1}-1\right)
  >3+(2r-1)\cdot 2^r\ge 3+j\cdot 2^{r}$ -- contradiction.
\end{proof}

\begin{proposition}
  For $k\ge 1$, $r\ge 4$, and $j\le 2r$ no projective $2^r$-divisible $\left[3+j\cdot 2^{r},k\right]$ code exists.  
\end{proposition}
\begin{proof}
  Due to Lemma~\ref{lemma_aux_1} it suffices to consider $j=2r$. The case $r=4$ is given by Theorem~\ref{main_thm}. For $r>4$ we proof the statement 
  by induction on $r$. Assuming the existence of such a code, Equation~(\ref{mw1}) minus $r2^r$ times Equation~(\ref{mw0}) yields
  \begin{equation}
    \label{ie_special_2}
    \sum_{i>0} (i-r)2^r \cdot a_{i2^r}=3\cdot 2^{k-1}+r\cdot 2^r> 0. 
  \end{equation}  
  The residual code of a codeword of weight $i2^r$ is projective, $2^{r-1}$-divisible, and has length $3+(2r-i)\cdot 2^r$. If $i\ge r+2$, then 
  we can apply Lemma~\ref{lemma_aux_1} to deduce $a_{i2^r}=0$. For $i=r+1$ the induction hypothesis gives $a_{i2^r}=0$. Since $(i-r)2^r \cdot a_{i2^r}\le 0$ 
  for $i\le r$ the left hand side of Inequality~(\ref{ie_special_2}) is non-positive -- contradiction. 
\end{proof}



\end{document}